\documentclass{article}
\usepackage{amsmath}
\usepackage{amssymb, amsthm, amsfonts, tikz-cd, mathrsfs,mathtools,stmaryrd,enumitem, algorithmic,float,comment,tikz}
\usetikzlibrary{backgrounds}
\usetikzlibrary{decorations.pathreplacing}
\usepackage[colorlinks,citecolor=blue,linkcolor=blue,urlcolor=red]{hyperref}
\usepackage{fullpage}
\usepackage{xcolor}
\usepackage[ruled,vlined,linesnumbered]{algorithm2e}
\usepackage{thm-restate}

\newcommand{\poly}{\textup{poly}}

\usepackage[capitalize]{cleveref}
\newtheoremstyle{dotless}{}{}{\itshape}{}{\bfseries}{}{ }{}
\theoremstyle{dotless}
\newtheorem{theorem}{Theorem}[section]
\newtheorem{lemma}[theorem]{Lemma}
\newtheorem{prop}[theorem]{Proposition}
\newtheorem{cor}[theorem]{Corollary}

\newtheoremstyle{dotlessdef}{}{}{\normalfont}{}{\bfseries}{}{ }{}
\theoremstyle{dotlessdef}

\newcommand{\Z}{\mathbb{Z}}

\newcommand{\lpr}[1]{\left(#1\right)}

\newcommand{\ceil}[1]{\lceil#1\rceil}

\newcommand{\ang}[1]{\langle#1\rangle}

\newcommand{\bin}{\{0,1\}}

\def\tO{\widetilde{O}}

\newcommand{\eps}{\varepsilon}

\newcommand{\dist}{\text{dist}}

\newcommand{\ints}[1]{\llbracket #1 \rrbracket}

\title{Simpler and Higher Lower Bounds for Shortcut Sets}

\author{Virginia Vassilevska Williams\thanks{Massachusetts Institute of Technology. \texttt{virgi@mit.edu}. Supported by NSF Grants CCF-2129139 and CCF-2330048 and BSF Grant 2020356.}
\and 
Yinzhan Xu\thanks{Massachusetts Institute of Technology. \texttt{xyzhan@mit.edu}. Partially supported by NSF Grants CCF-2129139 and CCF-2330048 and BSF Grant 2020356.}
\and
Zixuan Xu\thanks{Massachusetts Institute of Technology. \texttt{zixuanxu@mit.edu}.}}
\date{\today}

\begin{document}
\pagenumbering{gobble} 
\maketitle

\begin{abstract}
We study the well-known shortcut set problem: how much can one decrease the diameter of a directed graph by adding a small set of shortcuts from the transitive closure of the graph. 

We provide a variety of lower bounds. First, we vastly simplify the recent construction of Bodwin and Hoppenworth [FOCS 2023] which showed an $\widetilde{\Omega}(n^{1/4})$ lower bound for the diameter of a directed unweighted $n$-node graph after adding $O(n)$ shortcut edges. We highlight that our simplification completely removes the use of the convex sets by B{\'a}r{\'a}ny and Larman [Math. Ann. 1998] used in all previous lower bound constructions. Our simplification also removes the need for randomness and further removes some log factors. It allows us to generalize the construction to higher dimensions, which in turn can be used to show the following results:
\begin{itemize}
    \item There is an $\Omega(n^{1/5})$ lower bound for the diameter of the graph after adding $O(m)$ shortcuts, where $m$ denotes the number of edges in the input graph.
    \item  For all $\varepsilon > 0$, there exists a $\delta > 0$ such that there are $n$-vertex $O(n)$-edge graphs $G$ where adding any shortcut set of size $O(n^{2-\varepsilon})$ keeps the diameter of $G$ at $\Omega(n^\delta)$. This improves the sparsity of the constructed graph compared to a known similar result by Hesse [SODA 2003]. 
    \item For any integer $d \ge 2$, there exists a graph $G=(V, E)$ on $n$ vertices and $S \subseteq V$ with $|S| = \widetilde{\Theta}(n^{3/(d+3)})$, such that when adding  $O(n)$ or $O(m)$ shortcuts, the sourcewise  diameter (the largest distance from some vertex in $S$ to some reachable vertex in the graph) is $\widetilde{\Omega}(|S|^{1/3})$. This initiates the study of sourcewise diameter in the setting of the shortcut set problem; previously, the study of the sourcewise variant is popular in a wide variety of related problems such as spanners and distance preservers. Complementing this lower bound result, we also provide an upper bound: we show that, we can reduce the sourcewise diameter to $\widetilde{O}(\sqrt{|S|})$ by adding $O(n)$ shortcut edges. 
\end{itemize}
\end{abstract}

\newpage
\pagenumbering{arabic}

\section{Introduction}
In a variety of computational models, the complexity of shortest paths problems crucially depends on the {\em diameter} of the underlying graph: the largest shortest paths distance between reachable pairs. In dynamic graph data structures, for instance, the update time and query time can both depend on the diameter. In parallel algorithms, e.g. in the PRAM model, the span of many graph problems is a function of the diameter.

With these applications in mind, Thorup \cite{thorup_digraph} introduced the notion of a {\em shortcut set} of a graph $G=(V,E)$: a set of edges $H$ from the {\em transitive closure} of $G$ such that $(V,E\cup H)$ has small diameter.

By definition, shortcut sets preserve  all reachability information in $G$, yet can drastically reduce the diameter, thus significantly lowering the complexity of graph computations in many settings (e.g. \cite{AndoniSZ20,BernsteinGW20,CaoFR20,fineman2019nearly,HenzingerKN14,liu2019parallel,KarczmarzS21,KleinS97}).

Two main settings are of interest in the shortcut set literature: How much can one drop the diameter of an $n$-node, $m$-edge directed unweighted graph by adding a shortcut set of
(1) $O(n)$ edges? (2) $O(m)$ edges?

Both settings preserve the sparsity of the graph, within a constant factor.
Until recently, the only known shortcut set construction in either setting was as follows (see \cite{ullman1991high}). Given any $n$-node graph and any $D$, sample a set $S$ of $O(n \log n / D)$ vertices uniformly at random, and add all $\tO(n^2/D^2)$ edges\footnote{$\tO(t(n))$ is $O(t(n)\poly\log n)$ and $\widetilde{\Omega}(t(n))$ is $\Omega(t(n)/\poly\log(n))$.} from the transitive closure between pairs of vertices in $S$. It is easy to show that this achieves a diameter of $O(D)$. Setting $D$ to $\sqrt n$ shows that $\tO(n)$ shortcuts are sufficient to achieve diameter $O(\sqrt n)$, and setting $D$ to $n/\sqrt m$ shows that $\tO(m)$ shortcuts are sufficient to drop the diameter to $O(n/\sqrt m)$. The polylogs in the $\tO$ were removed by \cite{BermanRR10}.

In a recent breakthrough, Kogan and Parter \cite{koganparter} improved the above folklore construction. They showed that for any $D$, to drop the diameter to $\leq D$, it is sufficient to add a shortcut set of size  $\tO(n^2/D^3)$ if $D<n^{1/3}$ and $\tO((n/D)^{3/2})$ if $D\geq n^{1/3}$.

In the two most commonly studied scenarios, we get that
\begin{itemize}
\item with $\tO(n)$ shortcuts one can achieve diameter $O(n^{1/3})$, and
\item with $\tO(m)$ shortcuts one can achieve diameter $O((n^2/m)^{1/3})$.
\end{itemize}

Several papers achieved lower bound constructions. Improving upon Hesse \cite{Hesse03}, Huang and Pettie~\cite{Huang2018LowerBO} constructed a family of $n$-node, $m$-edge graphs in which one cannot drop the diameter under $\Omega(n^{1/6})$ using $O(n)$ shortcuts, and under $\Omega(n^{1/11})$ using $O(m)$ shortcuts. Lu, Vassilevska Williams, Wein, and Xu~\cite{Lu2021BetterLB} improved the second lower bound, showing that with $O(m)$ shortcuts one cannot drop the diameter below $\Omega(n^{1/8})$. Most recently, Bodwin and Hoppenworth \cite{BH23} improved the first lower bound as well, showing that with $O(n)$ shortcuts one cannot drop the diameter below $\widetilde{\Omega}(n^{1/4})$.

Bodwin and Hoppenworth's construction is in fact very sparse: the graphs have $n$ nodes and $O(n)$ edges. However, for technical reasons\footnote{Their basic construction only shows a lower bound of the diameter for adding $\frac{n}{\alpha \log n}$ shortcut edges for some constant $\alpha$. For the $O(n)$-shortcut setting, they have to apply a path subsampling argument which reduces the number of vertices in the graph to $n'$, so that $\frac{n}{\alpha \log n}$ becomes $c \cdot n'$ for arbitrarily large constant $c$ depending on the parameters for the subsampling. However,  the path subsampling argument does not necessarily reduce the number of edges. As a result, they obtain lower bounds for the $O(n)$-shortcut setting, but not for the $O(m)$-shortcut setting.} their construction does not imply an improvement for the $O(m)$-shortcut setting, where the $\Omega(n^{1/8})$ lower bound remains the best known.

Thus, using $O(n)$ shortcuts one can achieve diameter $O(n^{1/3})$ but cannot go below diameter $\widetilde{\Omega}(n^{1/4})$. Using $O(m)$ shortcuts one can achieve diameter $O((n^2/m)^{1/3})$ and cannot go below diameter $\Omega(n^{1/8})$.
The gap in the second case above is  particularly big, especially for small $m$. As $m$ becomes closer and closer to $n$, the upper bound exponent becomes closer and closer to $1/3$, whereas the lower bound exponent is only $1/8$. For the $O(n)$-shortcut setting the exponents $1/3$ and $1/4$ are much closer.

Thus a pressing question is:

\begin{center}
{\em Question: Can we reduce the gap between $\Omega(n^{1/8})$ and $O((n^2/m)^{1/3})$ for the diameter using $O(m)$ shortcuts?}
\end{center}

\subsection{Our Results}
The main contribution of our paper
 is a simplification and generalization of the Bodwin-Hoppenworth \cite{BH23} lower bound construction. In particular, our new lower bound construction is fully deterministic and unlike all previous lower bound constructions no longer relies on a result of \cite{BL98} on the number of integer points in a $d$-dimensional ball around the origin. Due to the simplicity of our construction, we are able to generalize it and obtain several new results.

 First, we recover the same lower bound as \cite{BH23} in the $O(n)$-shortcut setting; in fact we additionally remove some log factors.

\begin{restatable}{theorem}{ONRegime}
\label{thm:1/4}
For any $C > 0$, there exist infinitely many  $n$-vertex directed unweighted graph $G = (V,E)$ such that for shortcut set $H$ of size $|H| \le C n$, the graphs $G\cup H$ must have diameter $\Omega(n^{1/4})$.
\end{restatable}

The construction of Bodwin and Hoppenworth \cite{BH23} gives a diameter lower bound of $\Omega(n^{1/4}/\log^{9/4} n)$. Our construction gets rid off the $\log^{9/4} n$ loss. Similar to \cite{BH23}, our construction also extends to the setting where we are allowed to add $p$ shortcut edges for some other values of $p$. For $p \le n$, our bound is the same as theirs, up to $\tO(1)$ factors; for $p \gg n$, our bound is actually better than theirs, improving from $\widetilde{\Omega}(n^{5/4} / p)$ to $\Omega(n/p^{3/4})$. This was achieved by removing some small non-tightness in their analysis.

\begin{restatable}{theorem}{Tradeoff}
\label{thm:tradeoff}
For any $p \in [1, n^{4/3}]$, there exist infinitely many  $n$-vertex directed unweighted graphs $G = (V,E)$ such that for shortcut set $H$ of size $|H| \le p$, the graph $G\cup H$ must have diameter $\Omega(\frac{n}{p^{3/4}})$. 
\end{restatable}

Next, unlike \cite{BH23}, our new construction, coupled with a variant of the potential function argument used by \cite{BH23}, allows us to improve upon the $\Omega(n^{1/8})$ lower bound \cite{Lu2021BetterLB} in the $O(m)$-shortcut setting, addressing the main question from earlier:

\begin{restatable}{theorem}{OMRegime}
\label{thm:1/5}
     For any constant $C>0$, there exist infinitely many $n$-vertex $m$-edge directed unweighted graphs $G = (V,E)$ such that for shortcut set $H$ of size $|H|\le Cm$, the graph $G\cup H$ must have diameter $\Omega(n^{1/5})$.
\end{restatable}

In many works on spanners, emulators, distance preservers, hopsets and so on (e.g. \cite{CE,Kavitha17,KavithaV15,CyganGK13,KoganP22}), one often considers restrictions to the objects of study, where instead of addressing all pairs of vertices, only certain pairs are of interest. The two common settings are (a) the {\em pairwise} setting in which one is given a set $P\subseteq V\times V$ and one only cares about the paths or distances between node pairs in $P$, and (b) the {\em sourcewise} setting in which one is given $S\subseteq V$ and one only cares about pairs in $S\times V$.

Motivated by this, we consider the {\em sourcewise} setting for shortcut sets: given a graph $G=(V,E)$, a set $S\subseteq V$, if we add a set of edges $H$ from the transitive closure of $G$, how much can we decrease the sourcewise diameter of $G$, $\max_{(s, v) \in S \times V, \dist(s, v) < \infty} \dist(s,v)$.

In particular, we show that using $O(n)$ shortcuts, the sourcewise diameter can be reduced to $\tO(|S|^{1/2})$, following the method of Kogan and Parter \cite{koganparter}.
Additionally, for any $d \ge 2$, and $|S| = \widetilde{\Theta}(n^{3/(d+3)})$, the sourcewise diameter is $\widetilde{\Omega}(|S|^{1/3})$ by adding $O(n)$ or $O(m)$ shortcuts (via a randomized construction).

\section{Technical Overview}

\subsection{Previous Works}

Before giving an overview of our construction, we first summarize previous approaches to proving lower bounds for the shortcut set problem using $O(n)$ shortcuts. Prior to the work of Bodwin and Hoppenworth~\cite{BH23}, lower bound constructions for the $O(n)$-shortcut set problem including \cite{Hesse03} and \cite{Huang2018LowerBO} are all based on constructing a directed graph $G$ on $n$-vertices with a set of \emph{critical paths} $\Pi$ such that $|\Pi| = C\cdot n$ where $C$ can be set as an arbitrarily large constant. The critical path set contains paths of  length $\ell$ that are unique paths between their endpoints, and are \emph{pairwise edge-disjoint}. Then one would need to add at least one shortcut for each critical path in order to decrease the diameter to $o(\ell)$. Using this approach, Hesse \cite{Hesse03} obtained a lower bound of $\Omega(n^{1/17})$, and Huang and Pettie \cite{Huang2018LowerBO} obtained a lower bound of $\Omega(n^{1/6})$.

Recently, Bodwin and Hoppenworth \cite{BH23} improved the lower bound to $\widetilde{\Omega}(n^{1/4})$ by relaxing the property where paths in $\Pi$ are pairwise edge-disjoint to allow pairs of paths to overlap on \emph{polynomially many} edges. 
In order to better illustrate our improvements, we give  more details of their construction. The construction falls under the same geometric construction framework used in earlier works (e.g. \cite{Hesse03,CE, Lu2021BetterLB, BHspanner22}) for constructing lower bounds for shortcut sets and related problems. In this framework, every vertex of the graph corresponds to a lattice point in $\Z^d$ and every edge corresponds to a $d$-dimensional vector. The construction of the set of vectors corresponding to the edges crucially relies on building a \emph{convex set} $B_d(r)$ of vectors of length $O(r)$, taken as the vertices of the convex hull of the set of integer points in a ball of radius $r$ centered at the origin.  It is known from~\cite{BL98} that $|B_d(r)| = \Theta\lpr{r^{d\cdot \frac{d-1}{d+1}}}$ and $B_d(r)$ has the important property of being \emph{convex}, namely for distinct $\vec{b}_1, \vec{b}_2\in B_d(r)$, we have $\ang{\vec{b}_1,\vec{b}_2} < \ang{\vec{b}_1,\vec{b}_1}$.

The construction in \cite{BH23} builds on the framework for $d = 2$  with some  modifications that enabled them to allow polynomial overlaps between the paths in $\Pi$. We now give a high-level sketch of their construction: for sufficiently large positive integer $r > 0$, take the  convex set $B:= B_2(r^3)$ with size $|B| = \Theta(r^2)$ (but they only keep points in the first quadrant). Let $B = \{\vec{b}_1,\dots, \vec{b}_{|B|}\}$ where the vectors are ordered counterclockwise. The graph $G = (V,E)$ on $n$ vertices with critical paths set $\Pi$ is defined as follows.
\begin{itemize}
    \item[-] \textbf{Vertex Set.} $G$ is a layered graph with $\ell = \widetilde{\Theta}(r^2)$ layers and each layer corresponds to the grid $[\widetilde{\Theta}(r^3)]\times [\widetilde{\Theta}(r^3)]$.\footnote{For a nonnegative integer $N$, $[N]$ denotes $\{1, \ldots, N\}$} A vertex $v$ in layer $i\in [\ell]$ is denoted as $(i,x_1,x_2)$ where the first coordinate specifies that $v$ is in layer $i$ and the last two coordinates specifies that $v$ corresponds to grid point $(x_1, x_2)$. 
    \item[-] \textbf{Edge Set.} Define the set of edge vectors using $B$ as 
    \[W := \{\vec{w}_i :=  \vec{b}_{i+1} - \vec{b}_i \mid i\in [|B|-1]\}.\]
    All the edges in $G$ are between adjacent layers. For $i\in [\ell]$, define $E_i = \{\vec{w}_{\lambda_i}, (0,0)\}$ where $\vec{w}_{\lambda_i}$ is uniformly sampled from the set $W$ with $\lambda_i\in [|B|-1]$ as the set of edge vectors between layer $i$ and layer $i+1$. For a vertex $(i,x_1,x_2)$ in layer $i$,  connect it to the vertex $(i+1,x_1+e_1,x_2+e_2)$ for all $(e_1,e_2)\in E_i$ if $(x_1+e_1,x_2+e_2)$ stays inside the grid. 

    \item[-] \textbf{Critical Paths $\Pi$.} Define the set of directional vectors $D := B$ where $D = \{\vec{d}_1,\dots, \vec{d}_{|B|}\}$ and $\vec{d}_i := \vec{b}_i$. Let  the starting zone $S$ for the paths in $\Pi$ be the middle square whose side length is $1/3$ of the grid length. For each vertex $(1,x_1,x_2)$ in the first layer where $(x_1,x_2)\in S$ and for each directional vector $\vec{d}\in D$, they uniquely define a critical path $\pi$ in $\Pi$ as follows: $\pi$ starts from the vertex $(1,x_1,x_2)$ in layer $1$ and in layer $i$, it takes the edge corresponding to the vector $\arg\max_{\vec{e}\in E_i} \ang{\vec{e}, \vec{d}}$ (it can be shown that the choice is unique). Note that if the path falls off the grid, then the path terminates.
\end{itemize}

The resulting graph satisfies the following properties:
\begin{enumerate}
    \item $|V| = \ell\cdot \widetilde{\Theta}(r^3)\cdot \widetilde{\Theta}(r^3) = \widetilde{\Theta}(r^8)$, $|E| \le 2|V|$, $|\Pi| = \widetilde{\Theta}(n)$.
    \item $G$ has  $\ell = \widetilde{\Theta}(r^2) = \widetilde{\Theta}\lpr{n^{1/4}}$ layers.
    \item Each path in $\Pi$ is the unique path between its endpoints, and has length $\Theta(\ell)$ with constant probability. 
    \item  With high probability, for any path $\sigma: u\leadsto v$ of length $g$ in $G$, there are at most $\ell / g$ paths in $\Pi$ containing $\sigma$ as a subpath.
\end{enumerate}
Using a potential argument, in which the potential is defined as the total distances between the two endpoints of the paths in $\Pi$, this construction can give a lower bound of $\widetilde{\Omega}(n^{1/4})$ for shortcut sets of size $\le c_0\cdot \frac{n}{\log n}$ for some fixed constant $c_0$.  Using the so-called path subsampling argument, they extended the $\widetilde{\Omega}(n^{1/4})$ lower bound to $O(n)$-sized shortcut sets, which is a substantial improvement compared to the previous best lower bound of $\Omega(n^{1/6})$ given in \cite{Huang2018LowerBO}.

In the following, we discuss two limitations of \cite{BH23}.

\paragraph{The $O(m)$-shortcut setting.}

As mentioned earliler, aside from the $O(n)$-shortcut set problem, another intensely studied setting is the $O(m)$-shortcut set problem where $m$ is the number of edges in the input graph. An immediate question following from the construction in \cite{BH23} is whether this construction can give an improvement for the $O(m)$-shortcut setting. In particular, notice that in their construction, the graph has constant, and hence $m = \Theta(n)$, so one may think that in this case a shortcut set of size $O(n)$ is also a shortcut set of size $O(m)$, and thus their construction should directly imply an $\Omega(n^{1/4})$ lower bound in the $O(m)$-shortcut setting. However, this is not the case. 

First we highlight the fact that the construction in \cite{Huang2018LowerBO} has $|\Pi| = \Theta(n)$ where the constant can be set  arbitrarily large, but in \cite{BH23} the size of $\Pi$ is $c_0 \cdot \frac{n}{\log n}$ for some fixed constant $c_0$ that cannot be set  arbitrarily large. This means that the construction in \cite{BH23} given above would only give a lower bound for shortcut sets of size $\le c_0\cdot \frac{n}{\log n}$, so in order to extend the lower bound to shortcut sets of size $C\cdot n$ for arbitrarily large $C$, Bodwin and Hoppenworth \cite{BH23} use the so-called path subsampling argument: they
consider a smaller graph $G'$ obtained from $G$ by restricting to a random subset of vertices $V' \subseteq V$ and ``merging'' paths between these vertices with short lengths. Namely, they added the edges between all $u,v\in V'$ such that the distance between $u,v$ is at most $O(\log n)$. Thus $G'$ has fewer vertices compared to $G$ while having roughly the same diameter up to $\log n$ factors. A shortcut $H$ on $G'$ can also be viewed as a shortcut on $G$ as  $V(G')$ is a subset of $V(G)$. Furthermore, $n' := |V(G')| < n$, so when $|H| = c_0\cdot \frac{n}{\log n}$ for some, $|H|$ can be $C \cdot n'$ for arbitrarily large $C$, if $n'$ (which can be adjusted by the parameters in the subsampling argument) is small enough. 

Due to the nature of the operation described above, only the number of vertices in the graph is decreased, and the number of edges does not necessarily decrease.  This is the underlying barrier in obtaining lower bounds for $O(m)$-shortcut sets using the technique in \cite{BH23}.

\paragraph{Generalization to higher dimensions}

A natural follow-up question on the construction given in \cite{BH23} is whether it can be generalized to higher dimensions and give better bounds. Previous constructions in~\cite{Huang2018LowerBO, Lu2021BetterLB} can both be easily generalized to higher dimensions and one can optimize the dimension to obtain the best bounds. However, it is unclear how the construction in \cite{BH23} can be generalized to higher dimensions for the following reason: the set of edge vectors $W$ is taken to be the set of differences between adjacent vectors in $B$ ordered counterclockwise. Although it is easy to define the notion of ``adjacent vectors'' in $2$ dimensions, it is unclear how to define  an analogous notion in, say, $3$ dimensions.

\subsection{Our Construction}

The key observation that led to our construction is the following: notice that in the construction of the critical paths in \cite{BH23}, each path has two choices at each layer, which is to take the edge corresponding to the zero vector $(0,0)$ or to take the edge corresponding to the nonzero vector $\vec{w}_{\lambda_i}$. Thus, a critical path in fact corresponds to a subset of the set of edge vectors $W$, and the difference between the end point and start point of the path corresponds to sum of the vectors in this subset. The use of the  convex set $B$ in the construction of the set of edge vectors $W$ ensures that critical paths are well-defined in the sense that it has a well-defined choice between the two vectors in $E_i = \{\vec{w}_{\lambda_i}, (0,0)\}$ and that the sum of vectors taken in a critical path is the sum of a subset of $W$ whose sum is unique among all subsets, since the subset uniquely maximizes the inner product with the corresponding directional vector. Intuitively, this means that we only  need a set of edge vectors $W$ that have  lots of ``unique subset sums''.

There is in fact a very simple set that satisfies the above described property, which is the set $W := \{1\}\times [r] = \{(1,s)\mid 1\le s\le r\}$ for some integer $r$. 
Between layer $i$ and layer $i+1$, we put the set of edge vectors as $E_i = \{\vec{w}, (0,0)\}$ for some $\vec{w}\in W$ and each vector $\vec{w}\in W$ is assigned to a unique $E_i$, so we have $\ell = r+1$. Then we define the critical paths as the paths that take the nonzero edge vectors of the form $\{1\}\times [t]$ for $t\in [r]$ starting in some starting region in the first layer. Notice that  given a starting vertex $(0, x_1,x_2)$, a critical path taking the vectors $\{1\}\times [t]$ will end at the vertex $(r, x_1 + t, x_2 + \sum_{j = 1}^{t} j)$. It is easy to see that the start vertex and end vertex uniquely determine a path since the second coordinate indicates the number $t$ of nonzero vectors taken and the third coordinate indicates that the sum of these $t$ distinct numbers is $\sum_{j = 1}^{t} j$. The unique subset of $[r]$ satisfying these conditions is  $[t]$. This construction completely removes the use of the result from \cite{BL98} that bounds the number of vertices of the convex hull over the integer points contained in a $d$-dimensional ball of radius $r > 0$. In addition, we  eliminates the use of randomness in \cite{BH23} and the new construction is easier to analyze. 
We further remark that our simplified construction also provides a logarithmic improvement for the lower bound of $O(n)$-shortcut set over \cite{BH23}. 

In addition, this construction can be  easily generalized to higher dimensions. For instance, when each grid is $3$-dimensional, we take the set of edge vectors as $W:= \{1\}\times \{0\} \times [r]\cup \{0\}\times \{1\} \times [r]$ and let the critical paths correspond to subsets of the form $\{1\}\times \{0\}\times [s] \cup \{0\}\times \{1\}\times [s']$ for $s,s'\in [r]$. It is not hard to see that given the coordinates of the start vertex and the end vertex of a critical path, there is a unique path between the endpoints. Using such generalizations, we are able to obtain improvements in other settings such as the $O(m)$-shortcut set setting.

In our construction for the $O(m)$-shortcut set setting, we also use a refinement of the potential function used in \cite{BH23}. In their work, the potential function is defined as the sum of distances in the graph $G \cup H$ between the endpoints over all critical paths\footnote{We slightly abuse notation and use $G \cup H$ to denote $(V(G), E(G) \cup H)$. }, where $H$ is the added shortcut set. Our potential function is similar, but the distances are the distances in $G \cup E' \cup H$, where $E'$ consists of  pairs $(u, v)$ where the distance from $u$ to $v$ in $G$ is small. In our construction where each grid is $3$-dimensional, adding a shortcut between two vertices with smaller distance will drop the potential function defined in \cite{BH23} more. Thus, intuitively, adding $E'$ to our graph $G$ will make the potential drop  be $0$ for adding a shortcut between two vertices with very small distance. As a result, we can add more shortcuts before our potential function drops below a constant fraction of its initial value.

We will discuss our constructions in more detail in \cref{sec:2d} and \cref{sec:d-dim}.
\section{Preliminaries}

For nonnegative integer $N$, we use $[N]$ to denote $\{1, \ldots, N\}$, and $\ints{N}$ to denote $\{0, \ldots, N-1\}$.

We use $\log$ in this paper to denote log base $2$. 

For a graph $G=(V, E)$ and vertices $u, v \in V$, we use $\dist_G(u, v)$ to denote the distance from $u$ to $v$ in $G$. The subscript $G$ might be dropped if it is clear from context.  
For a path $\pi$, we use $s_\pi$ to denote its start point and $t_\pi$ to denote its end point.

\section{Base Construction}\label{sec:2d}

In this section, we give a simplified deterministic construction of a graph that gives an $\Omega(n^{1/4})$ lower bound in the $O(n)$-shortcut regime, matching the result in \cite{BH23}. 

\subsection{Construction of \texorpdfstring{$G = (V, E)$}{G = (V, E)} with Path Set \texorpdfstring{$P$}{P}}

First, we describe our graph construction. Let $r$ be an integer  and we may assume without loss of generality that $r$ is a power of $2$.

\paragraph{Vertex set.} 
The graph has $\ell = r+1$ layers indexed from $0$ to $r$ and each layer is the grid $\ints{2r} \times \ints{2r^2}$. In every layer, each vertex corresponds to a point $(x_1,x_2)\in \ints{2r}\times \ints{2r^2}$. We denote a vertex $v$ in layer $i$ as $(i, x_1,x_2)$ where the first coordinate specifies that $v$ is in layer $i$, and the last two coordinates $(x_1,x_2)$ specifies the grid point $v$ corresponds to.

\paragraph{Edge set.} 
Consider the set of integer vectors $W = \{1\}\times \ints{r}$. We only add edges between adjacent layers as follows: between layer $i$ and layer $i+1$ for $i = 0,\dots,r-1$, define the edge vector set $E_i = \{(w_1,w_2), (0,0)\}$ for some $(w_1,w_2)\in W$. (We will specify how to pick these $(w_1,w_2)\in W$ in details later.) For every vertex $(i, x_1,x_2)$, we add an edge connecting to the vertex $(i+1, x_1 + e_1, x_2+e_2)$ for every $(e_1,e_2)\in E_i$. Note that we will not add an edge if $(x_1+e_1, x_2+e_2)\notin \ints{2r}\times \ints{2r^2}$.

Now we describe how to assign the nonzero edge vector $(w_1,w_2)\in W$ to edge vector set $E_i$ for each adjacent layer $i\in \ints{r}$.  Let $q$ be a permutation on $\ints{r}$ such that the $(\log r)$-bit binary representation of $q_i$ is  exactly the binary representation of $i$ reversed (for simplicity, we use sets of integers indexed from $0$, since integers in $\ints{r}$ can all be represented using $(\log r)$ bits). See \cref{tab:p_example} for an example. Then $E_i = \{(1, q_i), (0, 0)\}$.

\begin{table}[ht]
    \centering
    \begin{tabular}{cccccccccccccccccc}
    \hline
    $i$ &  0 & 1 & 2 & 3 & 4 & 5 & 6 & 7 & 8 & 9 & 10 & 11 & 12 & 13 & 14 & 15\\
    \hline
    $q_i$ & 0 & 8 & 4 & 12 & 2 & 10 & 6 & 14 & 1 & 9 & 5 & 13 & 3 & 11 & 7 & 15\\
    \hline
    \end{tabular}
    \caption{Example of the permutation $q$ for $r=16$. }
    \label{tab:p_example}
\end{table}

\paragraph{Critical paths $P$.}

For any $(x_1, x_2) \in \ints{r} \times \ints{r^2}$ and for any $s \in [r]$, we define a critical path as follows. The starting point of the critical path is $(0, x_1, x_2)$. For any layer $i$, the critical path uses the edge corresponding to $(1, q_i)$ if $q_i < s$; otherwise, it uses the edge corresponding to $(0, 0)$. In particular, the set of nonzero vectors taken by a critical path is of the form $\{1\}\times \ints{s}$ for $s\in [r]$.

\paragraph{Properties of $G$.}

By definition of our construction, the graph $G = (V, E)$ with the critical path set $P$ has 
\begin{align*}
    |V| &= \ell\times (2r)(2r^2) = 4r^3(r+1),\\
    |E| &\le 2(\ell-1)(2r)(2r^2) = 8r^4, \\
    |P| &= r\cdot r^2 \cdot |D| = r^4.
\end{align*}

Next we show that all critical paths are  unique paths between their endpoints and they have relatively small pairwise overlaps.

\begin{lemma}\label{lem:2d-unique-path}
    For every critical path $\pi\in P$, $\pi$ is the unique path between its endpoints in $G$.
\end{lemma}

\begin{proof}
    Let $v = (1, x_1,x_2)$ be the start vertex of $\pi$ and suppose the subset of nonzero edge vectors used by $\pi$ is exactly the set $\{1\}\times \ints{s}$ where $s\in [r]$. Then we can compute the endpoint of $\pi$ as $u := (r+1, y_1,y_2)$ where $y_1 =x_1+s, y_2 = x_2 + \sum_{i = 0}^{s-1} = x_2 + s(s-1)/2$. 

    Suppose for contradiction that there exists another path $\pi'\ne \pi$ from $v$ to $u$ in $G$. Let $S$ denote the set of nonzero edge vectors in $\pi'$, then we know that since $\pi'\ne \pi$ and each adjacent layer corresponds to a unique nonzero edge vector, we must have $S\ne \{1\}\times \ints{s}$. Since all the nonzero edge vectors have the first coordinate equal to $1$ and we know the endpoints of $\pi'$ correspond to the grid point $(x_1,x_2)$ and $(x_1+s, x_2+s(s-1)/2)$, we know that $\pi'$ must take exactly $s$ nonzero edge vectors. Thus we must have $|S| = s$. Then since the second coordinates of all the edge vectors are exactly the set $\ints{r}$, in order for $\pi'$ to reach $u$, the sum of the second coordinates of the $s$ vectors in $S$ must be $s(s-1)/2$. It is easy to see that the only solution is to take the values $0,\dots, s-1$ for the second coordinates of the vectors in $S$ and therefore $S = \{1\}\times \ints{s}$. This gives a contradiction to the assumption that $\pi'\ne \pi$.
\end{proof}

\begin{lemma}\label{lem:same-d-vertex-disjoint}
    For any two critical paths $\pi_1,\pi_2\in P$ corresponding to the same subset of nonzero edge vectors $\{1\} \times \ints{s}$ (and thus having distinct starting vertices), we have that $\pi_1$ and $\pi_2$ are vertex-disjoint.
\end{lemma}

\begin{proof}
As $\pi_1,\pi_2$ have different start point, and at each layer, they will pick the same edge vector,  they will stay different at each layer. 
\end{proof}

\begin{lemma}\label{lem:2d-overlap}
    Given any path $\sigma: u \leadsto w$ of length $g$ in $G$ from layer $a$ to layer $b$ where $a,b\in \ints{\ell}$ and $b-a = g$, there are at most $8r / g$ critical paths containing $\sigma$.
\end{lemma}

\begin{proof}
    Let $k$ be the smallest integer such that $g\ge 2\cdot \frac{r}{2^k}$, so that $\{a, \ldots, b-1\}$ contains an interval $\{i \cdot \frac{r}{2^k}, i \cdot \frac{r}{2^k}+1, \ldots, i \cdot \frac{r}{2^k} + \frac{r}{2^k}-1\}$ for some $i \in \ints{2^k}$. Note that by definition of the permutation $q$, $$\{q_{i \cdot \frac{r}{2^k} }, \ldots, q_{i \cdot \frac{r}{2^k} + \frac{r}{2^k} - 1}\} = \{i', i' + 2^k, \ldots, i' + r - 2^k\}, $$
    where the $k$-bit binary representation of $i'$ is the $k$-bit binary representation of $i$ reversed. 
    
    Consider the set of nonzero edge vectors $S = \bigcup_{i = a}^{b-1} E_i \setminus \{(0,0)\}$. By the above discussion, $S$ must contain  edge vectors $\{(1,i'+ t \cdot 2^k) \mid t \in \ints{r/2^k}\}$.

    Suppose the path $\sigma$ is used by a critical path $\pi$ with initial point $(0, x_1, x_2)$ and with the set of edge vectors $\{1\} \times \ints{s}$. Now we want to bound the number of critical paths $\pi'\in P$ that can potentially contain $\sigma$ as a subpath. By \cref{lem:same-d-vertex-disjoint}, it suffices to count the number of $s' \in \ints{r}$ such that some critical path $\pi'$ passing through the vertex $u$ in layer $a$ associated with the set of nonzero edge vectors $\{1\} \times \ints{s'}$ can contain $\sigma$ as a subpath.
    
    We show that $\pi'$ cannot contain $\sigma$ as a subpath if   $|s' - s| > 2^k$. Notice that if $|s' - s| > 2^k$, then $S$ must contain an edge vector $w\in W$ that is in the symmetric difference of $\{1\}\times \ints{s}$ and $\{1\}\times \ints{s'}$ since $S$ contains the set $\{(1,i'+ t \cdot 2^k) \mid t \in \ints{r/2^k}\}$. This means that from layer $a$ to layer $b$, exactly one of $\pi$ and $\pi'$ takes the edge vector $w$, so $\pi$ and $\pi'$ must have taken different paths from layer $a$ to layer $b$.

    It follows that the number of critical paths containing $\sigma$ as a subpath is at most the number of $s'$ such that $|s' - s| \le 2^k$. By definition of $k$, we have $k = \ceil{\log \frac{2r}{g}} \le 1+ \log \frac{2r}{g}$. So $2^k \le 4r/g$ and thus the number of critical paths containing $\sigma$ as a subpath is at most $8r/g$ as desired.
\end{proof}

We summarize the properties of the graph constructed above as follows.

\begin{prop}\label{prop:2d-construction}
There exists a constant $c_0 > 0$, so that there exists infinitely many $n$-vertex directed unweighted graphs $G = (V, E)$ and path sets $P$ of size $|P| = c_0 n$ with the following properties:
    \begin{enumerate}
        \item \label{item:2d-construction-item1} $G$ has $\ell = \Theta(n^{1/4})$ layers.
        \item Each path in $P$ starts in layer $0$ and ends in layer $\ell-1$.
        \item Each path in $P$ is the unique path between its endpoints in $G$.
        \item \label{item:2d-construction-item4} For any path $\sigma$ of length $g$, there are at most $8\ell / g$ paths in $P$ that contain $\sigma$ as a subpath.
    \end{enumerate}
\end{prop}

\subsection{Application: \texorpdfstring{$O(n)$}{O(n)}-Shortcut Lower Bound}

From \cref{prop:2d-construction}, we can prove the following lower bound.

\begin{prop}\label{prop:1/4}
There exists an absolute constant $c_1 > 0$ and infinitely many  $n$-vertex directed unweighted graph $G = (V,E)$ such that for shortcut set $H$ of size $|H| \le c_1 n$, the graph $G\cup H$ must have diameter $\Omega(n^{1/4})$.
\end{prop}

\begin{proof}
    Take a graph $G = (V,E)$ from \cref{prop:2d-construction} with path set $P$ of size $|P| = 10 c_1 n$ for a suitable constant $c_1$. Let $H$ be a shortcut set of size $|H| \le c_1 n$ on $G$. 

    We define the following potential function over shortcut sets $H$ as 
    \[\phi(H) = \sum_{\pi \in P} \dist_{G\cup H}(s_\pi,t_\pi),\]
    i.e., $\phi(H)$ is the sum over the distances between pairs of endpoints of the critical paths after adding the shortcut set $H$. Then by \cref{prop:2d-construction} we have 
    \[\phi(\varnothing) = \sum_{\pi \in P}\dist_{G}(s_\pi,t_\pi) = 10 c_1 n\cdot (\ell - 1).\]
    Now for any path $\sigma: u\leadsto v$ of length $g$ in $G$, adding the shortcut $(u,v)$ only affects the critical paths $\pi$ that contains $\sigma$ as a subpath by \cref{lem:2d-unique-path}.
     For any shortcut set $H'$ and not including the shortcut $(u,v)$, let $\dist_1(\cdot, \cdot)$ denote the distance function $\dist_{G\cup H'}(\cdot, \cdot)$ and $\dist_2$ denote the distance function $\dist_{G\cup H' \cup \{(u,v)\}}$ for simplicity. Then we have
    \begin{align*}
     & \dist_1(s_\pi, t_\pi) - \dist_2(s_\pi, t_\pi) \\
      \le & \dist_1(s_\pi, t_\pi) - \min\left\{\dist_1(s_\pi, t_\pi), \dist_2(s_\pi, u) + \dist_2(u, v) + \dist_2(v, t_\pi) \right\}\\
      \le & \max\left\{0, \dist_1(s_\pi, t_\pi) - \left( \dist_2(s_\pi, u) + \dist_2(u, v) + \dist_2(v, t_\pi)\right) \right\}\\
      \le & \max\left\{0, \left( \dist_1(s_\pi, u) + \dist_1(u, v) + \dist_1(v, t_\pi) \right) - \left( \dist_2(s_\pi, u) + \dist_2(u, v) + \dist_2(v, t_\pi)\right) \right\}\\
      \le & \dist_1(u, v) - \dist_2(u, v)\\
      \le & g - 1 \le g.
    \end{align*}
    By \cref{item:2d-construction-item4} in \cref{prop:2d-construction}, adding the edge $(u,v)$ only affects the distance of $8\ell/g$ pairs of $(s_\pi, t_\pi)$ for $\pi\in P$, so we have for any shortcut set $H'$,
    \[\phi(H') - \phi(H'\cup \{(u,v)\})\le g\cdot \frac{8\ell}{g} = 8\ell.\]
    Therefore for a shortcut set $H$ of size $|H|\le c_1 n$, we have
    \[\phi(H) \ge \phi(\varnothing) - 8\ell c_1 n \ge 10(\ell - 1)c_1 n - 8\ell c_1 n = (2\ell - 10)c_1 n.\]
    Over $|P| = 10 c_1 n$ pairs of critical paths, the average distance is $\Theta(\ell)$, which implies that there must exists a pair with distance $\Theta(\ell)$. Then it follows that the diameter of $G\cup H$ is $\Theta(\ell) = \Theta(n^{1/4})$ by \cref{item:2d-construction-item1} in \cref{prop:2d-construction}.
\end{proof}

Similar to \cite{BH23}, \cref{prop:2d-construction} and \cref{prop:1/4} can be extended to show that, for every $p \in [1, c_1 n]$, there are $n$-node graphs for which adding $p$ shortcut edges keeps diameter as $\Omega(\frac{n}{p^{3/4}})$. The basic idea is to replace each vertex in the graph with a path of certain length, so that the number of vertices and number of layers increase, and $p$ remains the same. We will omit its proof for conciseness. For larger $p$, we can also obtain a lower bound by combining \cref{prop:1/4} and the path subsampling argument in \cite{BH23}. Directly applying their path subsampling argument, however, will incur the loss of $\log$ factors in the diameter lower bound and requires the use of randomization. Thus, we provide a different proof that is similar to the path subsampling argument. Additionally, our bound for superlinear $p$ is slightly better than that in \cite{BH23}, as they had a small non-tightness in their analysis.

\Tradeoff*

\begin{proof}
    Let $G = (V, E)$ be a graph from \cref{prop:1/4}, such that when adding $p = c_1 N$ shortcut edges, the diameter of $G$ is $\Omega(N^{1/4})$. Also, note that $G$ is a layered graph with $\Theta(N^{1/4})$ layers.  We will focus on the $ p > c_1 n$ case in the proof.

    We then create a graph $G'$ from $G$ as follows. Let $k = N / n$ and  let $V' \subseteq V$ be a subset of vertices that keep all vertices in every other $k$ layers so that $|V'| = n$. For every pair of vertices $u, v \in V'$, we add an edge from $u$ to $v$ in $G'$ if and only if $u$ can reach $v$ in $G$ and $u$ and $v$ are in two adjacent layers in $V'$. 
    
    Suppose for contradiction we can add $\le p$ shortcut edges $H$ so that the diameter of $G'$ becomes $o(\frac{n}{p^{3/4}}) = o(\frac{n}{N^{3/4}})$ (this would imply $N=o(n^{4/3})$, as otherwise, the diameter becomes $o(1) < 1$ which is impossible). Then we can add these same shortcut edges to $G$. For any pair of vertices $s, t \in V$ where $s$ can reach $t$, let $\pi$ be any path from $s$ to $t$. Then among the first $k$ vertices of $\pi$, there must be a vertex $s'$ in $V'$; similarly, among the last $k$ vertices of $\pi$, there must be a vertex $t'$ in $V'$. As $H$ reduces the diameter of $G'$ to $o(\frac{n}{N^{3/4}})$, and each edge in $G'$ corresponds to a path of length $k$ in $G$, the distance from $s'$  to $t'$ in $G \cup H$ must be $o(k \cdot \frac{n}{N^{3/4}}) = o(N^{1/4})$. Consequently, the distance from $s$ to $t$ in $G \cup H$ must be $\le 2k + o(N^{1/4}) = o(N^{1/4})$ ($k=o(N^{1/4})$  because $N=o(n^{4/3})$). This implies that the diameter of $G \cup H$ is at most $o(N^{1/4})$, contradicting  \cref{prop:1/4}.
\end{proof}

\cref{thm:1/4}, which we recall below, is an immediate corollary of \cref{thm:tradeoff} by setting $p = O(n)$:
\ONRegime*

\section{Higher Dimensional Construction}\label{sec:d-dim}

\subsection{Construction of \texorpdfstring{$G_d = (V_d, E_d)$}{Gd = (Vd, Ed)} with Path Set \texorpdfstring{$P_d$}{Pd}}

It is easy to generalize our construction to higher dimensions as follows. Let $d \ge 1$ be an integer and let $G_d = (V_d, E_d)$ be the graph defined as follows. Note that the dimension of the grid in each layer of the graph will be $d+1$ and that the base construction corresponds to the case where $d = 1$

\paragraph{Vertex set $V_d$} The graph has $\ell = d\cdot r + 1$ layers and each layer is the grid $\ints{2d\cdot r}^d \times \ints{2d\cdot r^2}$. We denote a vertex $v$ in layer $i$ as $(i, \vec{x})$ where $i$ specifies the layer and $\vec{x}\in \ints{2d\cdot r}^d \times \ints{2d\cdot r^2}$ is a grid point. 

\paragraph{Edge set $E_d$} Consider the set of integer vectors $W \subseteq \bin^d\times \ints{r}$, where the first $d$ coordinates of every vector in $W$ contains exactly a single one and $d-1$ zeros. Let $q$ be the same permutation on $\ints{r}$ as the one defined in \cref{sec:2d}. For each $k = 1,\dots, d$ and $i\in \ints{r}$, let $\vec{w}_{id+k-1}$ be the vector $(0,\dots, 0,1,0,\dots,0,q_i)$ where the $1$ is on the $k$-th coordinate. Note that $\vec{w}_0,\dots, \vec{w}_{dr-1}$ is a permutation of $W$.

We only add edges between adjacent layers as follows: between layer $i$ and $i+1$ for $i\in \ints{dr}$, define the edge vector set $E_i = \{\vec{w}_i, \vec{0}\}$ where $\vec{0}$ is the $(d+1)$-dimensional all zero vector. For every vertex $(i, \vec{x})$, we add edge connecting it to the vertex $(i + 1, \vec{x}+\vec{e})$ for $\vec{e}\in E_i$ if $\vec{x}+\vec{e} \in \ints{2d\cdot r}^d \times \ints{2d\cdot r^2}$. 

\paragraph{Critical Paths $P_d$} A critical path is defined by a start vertex $(0, \vec{x})$ where $\vec{x}\in \ints{r}^d \times \ints{r^2}$ and a set of vectors $\bigcup_{i = 1}^d \{0\}^{i-1}\times \{1\}\times \{0\}^{d-i}\times \ints{s_i}$ for $s_1,\dots, s_d\in [r]$. For any layer $i$, the critical path uses the edge corresponding to the nonzero vector $w_i$ if $w_i$ belongs to the set; otherwise, it uses the edge corresponding to $\vec{0}$. 

\paragraph{Properties of $G_d$} By definition, the graph $G_d = (V_d, E_d)$ has 
\begin{align*}
    |V_d| &= (\ell-1)\times (2dr)^d\times (2dr^2) = \Theta(r^{d+3}),\\
    |E_d| &\le 2 |V_d| = \Theta(r^{d+3}),\\
    |P_d| & = r^d\cdot r^2 \cdot r^d = \Theta(r^{2d+2}).
\end{align*}
It is not hard to check that $G_d$ also satisfies the following properties.
\begin{prop}[Properties of $G_d$]\label{prop:Gd-properties}
    The graph $G_d = (V_d, E_d)$ with critical path set $P_d$ satisfies the following:
    \begin{enumerate}
        \item \label{item:Gd-properties:unique-path} For every critical path $\pi\in P_d$, $\pi$ is the unique path between its endpoints in $G$.
        \item \label{item:Gd-properties:same-directions} For any two critical paths $\pi_1,\pi_2\in P_d$ associated with the same $s_1, \ldots, s_d$ (and thus have distinct starting vertices), we have that $\pi_1$ and $\pi_2$ are vertex-disjoint.
        \item \label{item:Gd-properties:overlap} Given any path $\sigma: u \leadsto v$ of length $g$ in $G_d$ from layer $a$ to layer $b$ where $a,b\in \ints{\ell}$ and $b-a = g$, there are at most $O((r / g)^d)$ critical paths containing $\sigma$.
    \end{enumerate}
\end{prop}
\begin{proof}\
    \begin{enumerate}
        \item Suppose the start vertex for $\pi$ is $u=(0, x_1, \ldots, x_d, y)$, and $s_1, \ldots, s_d$ are the associated values for $\pi$. Then the end point $v$ of $\pi$ can be computed as 
        $$\left(\ell - 1, x_1 + s_1, x_2 + s_2, \ldots, x_d + s_d, y + \sum_{i=1}^d \frac{s_i (s_i-1)}{2}\right).$$
        Thus, for any path travelling from $u$ to $v$, it must use $s_i$ nonzero vectors from $\{0\}^{i-1} \times \{1\} \times \{0\}^{d-i-1} \times \ints{r}$ for every $i \in [d]$. Furthermore, the sum of the last coordinate of these vectors must equal  $\sum_{i=1}^d \frac{s_i (s_i-1)}{2}$. As a result, for each $i$, it must use $\{0\}^{i-1} \times \{1\} \times \{0\}^{d-i-1} \times \ints{s_i}$, because any other choices of vectors would make the sum of the last coordinate larger than desired. Hence, there is a unique path from $u$ to $v$.
        \item As $\pi_1,\pi_2$ have different start vertices, and at each layer, they will pick the same direction vector because they share the same $s_1, \ldots, s_d$, so they will stay different at each layer. 
        \item Suppose the path $\sigma$ is used by a critical path $\pi$ with initial point $(0, \vec{x})$ and with associated values $s_1, \ldots, s_d$. Now we want to bound the number of critical paths $\pi'\in P$ that can potentially contain $\sigma$ as a subpath. By \cref{item:Gd-properties:same-directions}, it suffices to count the number of $s_1', \ldots, s_d' \in \ints{r}$ such that some critical path $\pi'$ passing through the vertex $u$ in layer $a$ associated with $s_1', \ldots, s_d'$ can contain $\sigma$ as a subpath.

        For every $k \in [d]$, $E_{a} \setminus \{\vec{0}\}, E_{a+1} \setminus \{\vec{0}\}, \ldots, E_{b-1} \setminus \{\vec{0}\}$ contain $\left\{ 
\{0\}^{k-1} \times \{1\} \times \{0\}^{d-k} \times \{q_i\} \mid \alpha \le i < \beta \right\}$ for some $\alpha, \beta$ where $\beta - \alpha = \Theta(g/d) = \Theta(g)$. By the same proof as \cref{lem:2d-overlap}, if $|s_k' - s_k|> \Omega(r/g)$, then $\pi'$ cannot contain $\sigma$ as a subpath. As this holds for every $k \in [d]$, we get that the number of $\pi'$ that can contain $\sigma$ as a subpath is $O((r/g)^d)$. 
    \end{enumerate}
\end{proof}

\subsection{Application I: \texorpdfstring{$O(m)$}{O(m)}-Shortcut Lower Bound}

For our $O(m)$-shortcut lower bound, we will use $d=2$ in the graphs constructed above.  Let us summarize the properties.

\begin{cor}\label{cor:d=2-construction}
For infinitely many $n$, there exists an $n$-vertex directed unweighted graphs $G = (V, E)$ and path sets $P$, where $|E| \le 2n$ and $|P| = \Theta(n^{6/5})$, and with the following properties:
    \begin{enumerate}
        \item $G$ has $\ell = \Theta(n^{1/5})$ layers.
        \item Each path in $P$ starts in layer $0$ and ends in layer $\ell-1$.
        \item \label{item: d=2-construction: unique-path} Each path in $P$ is the unique path between its endpoints in $G$.
        \item \label{item: d=2-construction: overlap} For any path $\sigma$ of length $g$, there are at most $O((\ell / g)^2)$ paths in $P$ that contain $\sigma$.
    \end{enumerate}
\end{cor}

From \cref{cor:d=2-construction}, we can prove the following lower bound.

\OMRegime*

\begin{proof}
    Take any graph $G = (V,E)$ with the associated critical paths $P$ from \cref{cor:d=2-construction}. 
    
    Let $c$ be a parameter to be chosen later. Let $E'$ a set of edges that contains $(u, v)$ to $E'$ if $\dist_G(u, v) \le c$. 

    Let $\Phi(H)$ be a potential function defined as 
    $$\Phi(H) = \sum_{\pi \in P} \dist_{G \cup E' \cup H} (s_\pi, t_\pi). $$
    Initially, $$\Phi(\emptyset) = |P| \cdot \frac{\ell}{c} = \Theta\left(\frac{n^{7/5}}{c}\right).$$

    Now, suppose we have a shortcut set $H'$, and we add  an edge $(u, v)$ to $H'$, where $u$ and $v$ are $g$ layers apart and $u$ can reach $v$ in $G$. For any $\pi \in P$, if $\pi$ does not have a subpath from $u$ to $v$, then adding $(u, v)$ will not affect the distance from $s_\pi$ to $t_\pi$, by \cref{item: d=2-construction: unique-path} in \cref{cor:d=2-construction}.

    If any critical path $\pi$ contains a subpath from $u$ to $v$, then there exists a unique path $\sigma$ from $u$ to $v$ in $G$, because otherwise, the critical path $\pi$ will not be the unique path from $s_\pi$ to $t_\pi$. By \cref{item: d=2-construction: overlap} in \cref{cor:d=2-construction}, the number of critical paths that contain $\sigma$ is $O((\ell/g)^2)$, which implies that the number of critical paths that contain a subpath from $u$ to $v$ is $O((\ell/g)^2)$. 

    Therefore, adding the edge $(u, v)$ will only affect the distance of $O((\ell/g)^2)$ pairs of $(s_\pi, t_\pi)$ for $\pi \in P$. For each such $\pi$ (for simplicity, we use $\dist_1$ to denote $\dist_{G \cup E' \cup H'}$ and $\dist_2$ to denote $\dist_{G \cup E' \cup H' \cup \{(u, v)\}}$ in the following), 
    \begin{align*}
    & \dist_1(s_\pi, t_\pi) - \dist_2(s_\pi, t_\pi) \\ 
    \le & \dist_1(s_\pi, t_\pi) - \min\left\{\dist_1(s_\pi, t_\pi), \dist_2(s_\pi, u) + \dist_2(u, w) + \dist_2(w, t_\pi) \right\}\\
    \le & \max\left\{0, \dist_1(s_\pi, t_\pi) - \left( \dist_2(s_\pi, u) + \dist_2(u, w) + \dist_2(w, t_\pi)\right) \right\}\\
     \le & \max\left\{0, \left( \dist_1(s_\pi, u) + \dist_1(u, w) + \dist_1(w, t_\pi) \right) - \left( \dist_2(s_\pi, u) + \dist_2(u, w) + \dist_2(w, t_\pi)\right) \right\}\\
     \le & \dist_1(u, w) - \dist_2(u, w)\\
     \le & \lceil g/c \rceil - 1. 
    \end{align*}
    Therefore, if $g \le c$, $\Phi(H') - \Phi(H' \cup \{(u, w)\}) = 0$; if $g > c$, $\Phi(H') - \Phi(H' \cup \{(u, w)\}) \le O((\ell / g)^2) \cdot (\lceil g/c \rceil - 1)=O(\frac{\ell^2}{gc}) = O(\frac{\ell^2}{c^2})$. 

    Thus, as $H$ contains at most $Cm \le 2Cn$ edges, 
    $$\Phi(H) \ge \Phi(\emptyset) - 2Cn \cdot O\left(\frac{\ell^2}{c^2}\right) = \Theta\left(\frac{n^{7/5}}{c}\right) - O\left(\frac{C n^{7/5}}{c^2}\right).$$
    Note that the constant factors hidden in $\Theta$ and $O$  above are independent of $C$ and $c$. Therefore, by setting $c$ big enough, we can ensure that $\Phi(H) \ge \Theta\left(\frac{n^{7/5}}{c}\right)$. By averaging, this implies that there exists $\pi \in P$ where $$\dist_{G \cup E' \cup H}(s_\pi, t_\pi) \ge \Omega\left(\frac{n^{1/5}}{c}\right).$$
    Also, as $\dist_{G \cup H}(s_\pi, t_\pi) \ge \dist_{G \cup E' \cup H}(s_\pi, t_\pi)$, we get that the diameter of $G \cup H$ is $\Omega\left(\frac{n^{1/5}}{c}\right)$. The theorem follows as $c$ is a constant only depending on $C$. 
\end{proof}

\subsection{Application II: Near-Quadratic-Shortcut Lower Bound}

Using the construction of $G_d$ (recall $|V_d| = \Theta(r^{d+3})$, $|P_d| = \Theta(r^{2d+2})$ and other properties in \cref{prop:Gd-properties}), we can show the following \cref{thm:quadratic-lb} that recovers \cite[Theorem 4.1]{Hesse03}. One small improvement is that the average degree of our construction is a constant, while the average degree of Hesse's construction depends on $\eps$. 

\begin{theorem}\label{thm:quadratic-lb}
    For all $\eps > 0$, there exists some $\delta > 0$ such that there exist infinitely many $n$-vertex $O(n)$-edge graphs $G$  where adding any shortcut set of size $O(n^{2-\eps})$ keeps the diameter of $G$ at $\Omega(n^\delta)$.
\end{theorem}

\begin{proof}
Let $d$ be a large enough integer so that $(1-\eps) \cdot \frac{d+3}{d-1} < 1 - \eps / 2$. Similar to the proof of \cref{thm:1/5}, we let $E'$ be a set of edges that contain $(u, v)$  if $\dist_{G_d}(u, v) \le c$, for some parameter $c$. 

Let $H$ be the shortcut set, and let $\Phi(H)$ be a potential function defined as 
    $$\Phi(H) = \sum_{\pi \in P_d} \dist_{G \cup E' \cup H} (s_\pi, t_\pi). $$

Initially, $$\Phi(\emptyset) = |P_d| \cdot \frac{\ell}{c} = \Theta\left(\frac{r^{2d+3}}{c}\right).$$

If we add an edge between two vertices that are $g \ge c$ layers apart, the potential function is dropped by at most 
$$O\lpr{(r/g)^d \cdot \frac{g}{c}} = O\left(\frac{r^d}{g^{d-1} c} \right) \le O\left(\frac{r^d}{c^{d}} \right).$$

Thus, if we add $O(r^{d+3} \cdot c^{d-1})$ edges, the potential is still $\Omega\left(\frac{r^{2d+3}}{c}\right)$, so the diameter of the graph is still $\Omega(r/c)$. By setting $c = r^{(1-\eps) \cdot \frac{d+3}{d-1}}$, we conclude  that after adding $$O(r^{d+3} \cdot c^{d-1}) = O(r^{d+3} \cdot r^{(1-\eps) \cdot (d+3)}) = O(n^{2-\eps})$$
edges, the diameter is still 
$$\Omega\lpr{\frac{r}{c}} = \Omega\lpr{\frac{r}{r^{(1-\eps) \cdot \frac{d+3}{d-1}}}} = \Omega(r^{\eps/2}) = \Omega(n^{\eps / (2(d+3))}).$$
\end{proof}

\section{Sourcewise \texorpdfstring{$O(n)$}{O(n)}-Shortcut Set}

Recall that in the sourcewise $O(n)$-shortcut set problem, we need to add $O(n)$ shortcut edges to a directed unweighted graph, so that $\max_{(s, v) \in S \times V, \dist(s, v) < \infty} \dist(s, v)$ is minimized. Similar to the $O(n)$-shortcut set problem, we are interested in how small the quantity could be. We provide both an upper and a lower bound for this problem. The upper bound adapts the techniques of Kogan  and  Parter~\cite{koganparter}, and the lower bound uses our higher dimensional construction in \cref{sec:d-dim}. 

\subsection{Lower Bound}

\begin{theorem}
\label{thm:SV-lower-bound}
For any integer $d \ge 2$, there exists a graph $G=(V, E)$ on $n$ vertices and $S \subseteq V$ with $|S| = \widetilde{\Theta}(n^{3/(d+3)})$, such that when adding  $O(n)$ (or $O(m)$) shortcuts, the sourcewise diameter is $\Omega(n^{1/(d+3)})$. 
\end{theorem}
\begin{proof}
In $G_d$, we independently put each vertex in the first layer in $S$ with probability $q = \frac{\log^2 n}{r^{d-1}}$. Recall that the number of vertices in each layer of $G_d$ is $\Theta(r^{d+2})$, 
so with high probability, $|S| = \Theta(r^{d+2} \cdot q) = \Theta(r^3 \cdot \log^2 n) = \Theta(n^{3/(d+3)}  \cdot \log^2 n)$. 

Similar to the proof of \cref{thm:1/5}, let $c$ be a parameter to be chosen later. Let $E'$ be a set of edges that contains $(u, v)$  if $\dist_G(u, v) \le c$. Additionally, we add $(u, v)$ to $E'$ if $u$ can reach $v$ and $v$ is in the first $\ell / 2$ layers. 

Let the potential function $\Phi(H)$ for a set of shortcut edges $H$ be defined  as 
$$\Phi(H) = \sum_{\pi \in P_d, s_\pi \in S} \dist_{G_d \cup E' \cup H}(s_\pi, t_\pi), $$
where we recall that $P_d$ is the set of critical paths in $G_d$. Because of the way we randomly pick $S$, the number of critical paths $\pi$ with $s_\pi \in S$ is $\Theta(|P_d| \cdot q) = \Theta(r^{d+3} \cdot \log^2 n)$ (recall $|P_d| = \Theta(r^{2d+2})$). Therefore, initially, $\Phi(\emptyset) = \Theta(r^{d+4} \cdot \log^2 n / c)$.

Now consider adding an edge $(u, v)$ to $H$  where $u$ can reach $v$ in $G$. Since if $\dist_{G_d}(u, v) \le c$, or if $v$ is in the first $\ell / 2$ layers, adding it will not change the value of the potential function, we assume that $\dist_{G_d}(u, v) = g > c$ and $v$ is in the second half of layers. First, by \cref{prop:Gd-properties}, the number of critical paths in $P$ using the subpath from $u$ to $v$ is at most $O((r/g)^d)$. Furthermore, for each vertex $s$ in the first layer, the number of critical paths starting from $s$ and at the same time using the subpath from $u$ to $v$ must be using the $s\leadsto v$ subpath in $G_d$. As a result, for each fixed $s$, there are $O((r/\dist_{G_d}(s, v))^d) = O(1)$ critical paths starting from $s$ and using the subpath from $u$ to $v$. Then by the random sampling of $S$, with high probability, the number of critical paths $\pi$ in $P$ using the subpath from $u$ to $v$ with $s_\pi \in S$ is 
$$O(\max\{\log n, (r/g)^d \cdot q\}).$$
This implies that, with high probability, the drop of potential caused by adding $(u, v)$ is 
$$O(\max\{g/c \cdot \log n, g/c \cdot (r/g)^d \cdot q\}).$$
Because $c \le g \le r$, the above can be bounded by 
$$O\left(\max\left\{r \log n / c, \frac{r \log^2 n}{c^{d}}\right\}\right).$$
If we add $O(n)$ edges to $H$, then the total potential drop can be bounded by $$O\left(\max\left\{r^{d+4} \log n/c, \frac{r^{d+4} \log^2 n}{c^{d}}\right\}\right).$$
By taking large enough $n$ and  setting $c$ as a large enough constant, we can make it so that the above is less than $\Phi(\emptyset) / 2$. Therefore, $\Phi(H) = \Theta(r^{d+4} \log^2 n / c)$ for $H=O(n)$. By averaging, some critical path $\pi$ with $s_\pi \in S$ will have length $\Omega(r / c)$. 

The same proof works for the $O(m)$-shortcut case because $G_d$ has $m = \Theta(n)$. 
\end{proof}

\subsection{Upper Bound}

We can assume without loss of generality that the input graph $G$ is a directed acyclic graph (DAG) by the following  trick used in prior works (e.g., see \cite{raskhodnikova2010transitive}):
for each strongly connected component, we  add a star connecting every vertex from and to the same vertex, which would reduce the diameter of each strongly connected component to $\le 2$ using only $O(n)$ shortcuts in total. Then we only need to keep one representative vertex in each strongly connected component, and the resulting graph would become a DAG. 

\begin{lemma}
    Given a graph $G = (V, E)$ and a subset $S\subseteq V$, one can reduce the sourcewise diameter to $\tO(\sqrt{|S|})$ using $O(n)$ shortcuts.
\end{lemma}

\begin{proof}[Proof Sketch]
The proof of this lemma is analogous to the shortcut set construction in \cite{koganparter}, so we only provide a proof sketch. 

Let $D = \sqrt{|S|}$ be our desired diameter. 
First, given a DAG $G$, let $G^*$ be its transitive closure. It is known that (e.g., see \cite{conf/soda/GrandoniILPU21}) for every $\ell$, we can decompose the vertices of any DAG to $\ell$ chains (a chain is a path), and $2n / \ell$ antichains (an antichain is a set of vertices that do not have edges between each other). Thus, we can decompose $G^*$ into $\ell := \frac{16 n}{D}$ chains $\mathcal{P} = \{P_1, \ldots, P_\ell\}$ and $2n/\ell$ antichains. Then for each $P_i$, we use known method \cite{raskhodnikova2010transitive} that decreases the diameter of it to $\le 2$ by adding $O(|P_i| \log |P_i|)$ shortcut edges. Next, let $\mathcal{P}'$ be a set of chains, and we add $P_i$ to $\mathcal{P}'$ for every $i \in [\ell]$ with probability $\Theta(\log n / D)$. Finally, for every $s \in S$ and every $P_i \in \mathcal{P}'$, we add an edge $(s, u)$ to the shortcut set where $u$ is the first vertex on $P_i$ that is reachable from $s$ (this step is different from \cite{koganparter}, since here we add edges from vertices in $S$, while in \cite{koganparter}, they add edges from a random subset of vertices). 

It is not difficult to see that, with high probability, the number of shortcut edges added is $\tO(n + \frac{n}{D^2}\cdot |S|) = \tO(n)$. The correctness follows by the analysis in \cite{koganparter}.

To decrease the number of shortcut edges added to $O(n)$, we perform the following two modifications:
\begin{itemize}
    \item When we add $O(|P_i| \log |P_i|)$ shortcut edges to $P_i$ to make its diameter $\le 2$, we instead only add $O(|P_i|)$ edges to make its diameter $\le \tO(1)$. This can be done by applying the same method to the chain formed by $\frac{|P_i|}{\log |P_i|}$ evenly spaced-out vertices  on the original chain $P_i$. 
    \item Slightly increase $D$ to $\tO(\sqrt{|S|})$ so that the $\tO(\frac{n}{D^2} \cdot |S|)$ part of the edge bound becomes $O(n)$. 
\end{itemize}
\end{proof}

\bibliographystyle{alpha}
\bibliography{ref}

\end{document}